\documentclass[thm-restate]{lipics-v2021}

\hideLIPIcs
\nolinenumbers

\usepackage{amsmath}
\usepackage{amsfonts}
\usepackage{mathtools}
\usepackage{graphicx}
\usepackage{color}

\usepackage[linesnumbered,lined,commentsnumbered,noend,boxed]{algorithm2e}

\usepackage{booktabs}

\ifpdftex 
    \usepackage[T1]{fontenc} 
\fi
\usepackage{lmodern}  
\usepackage{romanbar} 

\selectfont

\DeclareEmphSequence{\slshape,\itshape}

\usepackage{microtype}           
\DisableLigatures[-]{family=tt*} 
\usepackage{csquotes}            

\usepackage[capitalise,nameinlink,noabbrev]{cleveref}
\crefformat{equation}{#2(#1)#3}
\Crefformat{equation}{#2(#1)#3}

\usepackage[textsize=tiny, color=cyan]{todonotes}

\newcommand{\hquad}{\hspace{0.5em}} 
\newcommand{\Hquad}{\hquad}
\newcommand{\heq}{\hquad=\hquad}
\newcommand{\hle}{\hquad\le\hquad}
\newcommand{\hge}{\hquad\ge\hquad}
\newcommand{\hop}[1]{\hquad#1\hquad}
\newcommand{\qeq}{\quad=\quad}
\newcommand{\qle}{\quad\le\quad}
\newcommand{\qge}{\quad\ge\quad}

\usepackage{xspace}
\newcommand*{\algo}[1]{\texttt{\textup{#1}}\xspace}
\newcommand*{\sbcseries}{\sffamily\fontseries{sbc}\selectfont} 
\newcommand*{\sfefamily}{\fontfamily{lmssq}\selectfont}
\DeclareTextFontCommand{\textsbc}{\sbcseries}
\DeclareTextFontCommand{\textsfe}{\sfefamily}

\newcommand*{\mfrac}[2]{{\text{\footnotesize $\frac{#1}{#2}$}}}

\title{Graph Scheduling with Group Completion Times}

\author{Lars Rohwedder}{University of Southern Denmark, Odense, Denmark}{rohwedder@imada.sdu.dk}{https://orcid.org/0000-0002-9434-4589}{Supported by Dutch Research Council (NWO) project “The Twilight Zone of Efficiency: Optimality of Quasi-Polynomial Time Algorithms” [grant number OCEN.W.21.268]}

\author{Leander Schnaars}{Technical University of Munich, Munich, Germany}{leander.schnaars@tum.de}{https://orcid.org/0009-0004-6658-9150}{Supported by the Deutsche Forschungsgemeinschaft (DFG, German Research Foundation) - GRK 2201/2 - Projektnummer 277991500}

\authorrunning{L. Rohwedder and L. Schnaars}



\Copyright{Lars Rohwedder, Leander Schnaars}

\date{}

\begin{document}

\maketitle

\begin{abstract}
In the Graph Scheduling problem we schedule a given multiset of edges on discrete time steps, such that at each step the set of edges forms a matching. The goal is to minimize the sum of weighted group completion times, where a group is a set of edges and it completes when the last edge has been scheduled. Two popular variants of this problem are Coflow Scheduling and Data Migration. Our main result is extending a recent iterated rounding approach from Coflow Scheduling, roughly corresponding to the bipartite case, to the general Graph Scheduling problem. This yields an essentially tight $(2+\epsilon)$-approximation for the asymptotic setting where $\algo{OPT}$ is assumed to be large. For this we rely on polyhedral techniques from general matching, namely odd-set inequalities, and graph theoretical results on edge colorings in multigraphs. The state-of-the-art approximation algorithm for Data Migration is a $(1 + \phi)$-approximation that improves when $\algo{OPT}$ is small.
Taking the best of this and our main result, we obtain an improvement of the approximation rate for Data Migration in any regime.
\end{abstract}

\section{Introduction}
In Graph Scheduling with group completion times, the input consists of a multigraph $G = (V,E)$ and jobs $J_1,\dots,J_B \subseteq E$ with weights $\omega_1,\dots,\omega_B \in \mathbb{R}_{\ge 0}$. In other words, every job consists of some set of edges, where these sets do not necessarily have to be disjoint. In each discrete time step, we are allowed to schedule a matching on the graph, that is to say a set of edges not sharing any vertex. A job is completed if all of its edges have been scheduled. The goal is to minimize the weighted sum of completion times. 
Group completion times can model many common objectives. If for example there is a single job containing every edge, then the problem is equivalent to makespan minimization, while if every edge belongs to a unique job, this represents min-sum edge coloring. The single job case is equivalent to edge coloring and thus can be solved in polynomial time on bipartite graphs using Kőnig's Theorem \cite{König1916}, while the general case is $\algo{NP}$-hard \cite{holyer1981_edgecolorhard}. For min-sum edge coloring, the best known approximation ratio for the general case is a $1.8661$-approximation shown by Halldórsson, Kortsatz, and Sviridenko \cite{Halldorsson2008-msec}. On bipartite graphs, the best known ratios are $\sqrt{2}$ by Gandhi and Mestre \cite{Gandhi2009-edgecolor} for the unweighted case and $1.796$ by Halldórsson, Kortsatz, and Shachnai \cite{Halldorsson2003-edgecolor} for the weighted case.

Graph Scheduling has been studied mainly in the context of Coflow Scheduling \cite{chowdhury_efficient_2014,ahmadi_scheduling_2017,agarwal_sincronia_2018,khuller_select_2019,fukunaga22,rohwedder2025_coflow} and Data Migration \cite{mestre2010,Gandhi2013_datamigration}. In Coflow Scheduling, the underlying graph is required to be bipartite and the groups are formed by a disjoint partition of the edge set. This models a data exchange task commonly found in applications such as distributed computing, where between computation stages the different servers have to exchange intermediate results. Most common distributed computing frameworks such as MapReduce, Spark, and Hadoop \cite{Dean2008, Spark2010, Hadoop2010} contain algorithms for such tasks. The problem has thus enjoyed wide attention from both a theoretical as well as an applied side. For several years, the best known guarantees were multiple different $4$-approximations \cite{ahmadi_scheduling_2017,Shafiee2018-ik, agarwal_sincronia_2018, fukunaga22}, while there is a factor $(2-\epsilon)$-approximation hardness assuming $\algo{P} \neq \algo{NP}$ \cite{Sachdeva2013}. Recently, the factor $4$ has been improved to $3.415$ in \cite{rohwedder2025_coflow}. Some authors have studied extensions to the online setting \cite{khuller_select_2019} and to a variant with matroids instead of matchings in each time step \cite{im19}.

In Data Migration the graph is unrestricted, but each job consists of exactly the edges adjacent to some vertex, which implies that every edge belongs to exactly two jobs. As the name implies, this problem models a general data migration setting in a network, where every vertex represents a server and each edge is a data transmission requirement. A server finishes its data transmission task if all adjacent edges have been served. The best known approximation ratio for this setting is $(1+\phi) \approx 2.618$ \cite{mestre2010}. There is an extension of the problem to arbitrary demands on each edge for which a $4.96$-approximation is known \cite{Gandhi2013_datamigration}, though this changes the setting in a fundamental way and thus requires completely different approaches.

As remarked by Fukunaga \cite{fukunaga22}, some of the $4$-approximation algorithms for Coflow Scheduling \cite{ahmadi_scheduling_2017,Shafiee2018-ik} also work for general graphs, hence the best so far known approximation ratio for Graph Scheduling is $4$. In \cite{rohwedder2025_coflow}, Rohwedder and Schnaars introduce iterative rounding techniques that achieve an essentially tight $(2+\epsilon)$-approximation for Coflow Scheduling in an asymptotic regime where $\algo{OPT} \gg \sum_{j \in [B]}\omega_j$. Our main result extends this to general Graph Scheduling:
\begin{restatable}{theorem}{mainschedulingalg}\label{thm:main_scheduling_alg}
Given a Graph Scheduling instance $G$ with cost vector $\omega \in \mathbb{R}^B_{\ge 0}$ and parameter $\epsilon > 0$, in polynomial time we can find a solution $S$ such that:
\begin{equation*}
    \omega(S) \qle(2+\epsilon)\cdot \algo{OPT} \hop{+} \mathcal{O}\Big(\mfrac{1}{\epsilon}\Big)\cdot \sum_{j \in [B]}\omega_j
\end{equation*}
\end{restatable}
For instances with $\algo{OPT} \gg \sum_{j \in [B]}\omega_j$, this gives a $(2+ \epsilon)$-approximation, which is essentially tight due to the $(2-\epsilon)$ lower bound. This form of approximation guarantee is not often seen in the literature, but it turns out that many existing algorithms actually have \emph{negative} additive terms in their guarantees when analyzing them more precisely. In \cite{rohwedder2025_coflow}, an approximation framework is presented which allows obtaining approximation ratios from combinations of multiple such algorithms' allocation guarantees. Their framework can be directly lifted to the more general setting discussed in this work, but we do not employ it here and show the relevant results directly. Two examples in which negative additive terms can be derived are the greedy allocation for Coflow Scheduling in \cite{rohwedder2025_coflow} and the local ratio procedure for Data Migration by \cite{mestre2010}, which we refine in \cref{sec:data_migration}. These multiplicative and positive and negative additive guarantees can then be combined to obtain improved approximation algorithms only containing a multiplicative term.

In \cref{sec:data_migration} we use the refined analysis of local ratio in combination with \cref{thm:main_scheduling_alg} to obtain such an improved approximation ratio for Data Migration:

\begin{restatable}{theorem}{dmapproxresult}\label{theorem:main_dm_result}
There is a $2.617$-approximation algorithm for Data Migration with unit sized edges.
\end{restatable}

While the main focus is the extension of the techniques from Coflow Scheduling to general graphs and other related settings, from \cref{thm:main_scheduling_alg} combined with known results for Graph Scheduling, an improved approximation algorithm for the general case also follows. See \cref{sec:improved_graph_sched_approx} for full details.

\begin{restatable}{theorem}{graphschedulingapprox}\label{theorem:graph_scheduling_approx}
There is a polynomial time $\frac{1104}{281} < 3.93$-approximation algorithm for Graph Scheduling.
\end{restatable}

Other variants of Graph Scheduling have been studied, such as when instead of unit size, every edge has some transmission time $p_e \in \mathbb{R}_{\ge 0}$. However, such as in the case of Data Migration, this usually fundamentally changes the nature of the problem and thus also of the required approaches and obtained approximation guarantees. Note that this is distinct from the setting where edges have some multiplicities $p_e \in \mathbb{N}$, for which most results can be transferred at no or only a minor loss in approximation.
We also refer to \cite{megow2025groupschedule} for a broader study of scheduling problems with group completion times, where Graph Scheduling is one of the main applications.

\subsection{Techniques from Coflow Scheduling}
As this work is a generalization of the algorithms and techniques presented in \cite{rohwedder2025_coflow} for Coflow Scheduling, we briefly review the approach and overview the most important technical aspects. The approach has two core steps: In the first step, a deadline is determined for each job. These deadlines are chosen in a way to both fulfill a cost approximation guarantee while also guaranteeing feasibility of an allocation linear program. In the second step, through an iterative rounding procedure an integral point is obtained from the feasible LP. This integral solution does not directly correspond to a feasible schedule, as edges are only assigned to so called blocks and not to individual time slots. To obtain a valid schedule, a time slot assignment is computed using Kőnig's Theorem \cite{König1916} for bipartite graphs.

The algorithm presented in this work follows the same structure, but due to the general multigraph setting both the LP as well as the allocation algorithm have to be significantly modified. The first step of determining the deadlines works essentially unmodified. The second step however requires larger technical modifications, as we employ results from graph theory for which we have to structurally strengthen the LP. This added structure improves the obtained guarantees, but requires additional care when iteratively rounding.

\subsection{Extension to General Graphs}\label{sec:intro:ggextension}
Kőnig's Theorem guarantees the existence and polynomial time computability of a $\Delta(G)$ edge coloring of bipartite graphs of maximum degree $\Delta(G)$, which is an important subroutine of the algorithm. A $k$ edge coloring of a graph directly corresponds to a decomposition into $k$ matchings, so this provides an important tool for creating graph schedules. For general simple graphs, a similar guarantee is given by Vizing's Theorem \cite{Vizing1964}, which gives computability of a $\Delta(G)+1$ coloring in this case. However, for multigraphs, the best possible guarantees are much worse: By a theorem from Shannon \cite{Shannon1949}, every multigraph can be colored with $\lfloor\frac{3}{2}\Delta\rfloor$ colors and there are instances for which this bound is tight.
So if one were to use the approach from \cite{rohwedder2025_coflow} unmodified for general graphs, in the final step in which the schedule is created from the point obtained through iterative rounding, an additional approximation penalty of roughly $\frac{3}{2}$ might be incurred. While the overall approximation ratio might not worsen by as much due to the algorithm employing a combination of different allocation procedures, it is unlikely that large improvements over the approximation ratio of $4$ could be achieved and especially an asymptotic $(2+\epsilon)$-approximation would not be possible. 

To alleviate this, we strengthen the LP used in \cite{rohwedder2025_coflow} using exponentially many odd-set inequalities from Edmonds \cite{edmonds1965maximum} well known matching polytope. After rounding, the resulting integer solution has stronger properties, with which we can apply recent results from edge coloring theory to construct schedules using much fewer time slots than required by naive application of Shannon's algorithm. However, the added constraints create challenges in the iterative rounding, as iterative rounding is very sensitive towards the number of constraints, so more sophisticated counting arguments are required. The improved LP structure allows us to obtain algorithms which both achieve a strictly better than $4$-approximation for general Graph Scheduling as well as a $(2+\epsilon)$-approximation for Graph Scheduling in the asymptotic case. The approximation improvements for Data Migration also require the enhanced algorithm, as a direct approach does not yield sufficiently strong guarantees.

\section{LP Rounding}
In this section we describe a linear program for group completion times on graphs. This LP cannot be directly constructed from the underlying problem instance but rather needs suitable deadlines $C_1,\dots, C_B$ for the jobs $J_1,\dots,J_B$ as a prerequisite. In applications such as seen in \cref{sec:scheduling_algorithm,sec:data_migration}, these deadlines would for example be obtained through another algorithm as a first step. The deadlines have to be chosen such that the LP is feasible and together they describe certain combinatorial structure which fractional solutions of the underlying problem instance are required to obey. Instead of assigning edges to individual time slots, we group the slots into so called blocks, each of which consists of the time frame between two consecutive deadlines. 
We show how to iteratively round this LP to obtain an integral point and through this an integral assignment of edges to the blocks. The thus derived point might not be feasible for the original LP, but we establish strong bounds on the maximum constraint violations incurred through this rounding procedure. In a later step of the algorithm, the edge to block assignments are then turned into time slot assignments and thus into valid schedules for the underlying graph scheduling instance.

\subsection{Linear Program for General Graphs}
Assume that we have $B$ jobs and the corresponding deadlines $C_1,\dots,C_B$ for which the following LP is feasible. In the definition we assume without loss of generality that the jobs are ordered such that $C_1 \le C_2 \le \cdots \le C_B$ and we define $C_0 := 0$ for easier notation.

\begin{gather*}
\begin{aligned}
\sum_{b \in [B]}x_{e,b} \quad &= \quad 1 &\forall e \in E\\
\sum_{e: v \in e} x_{e,b} \quad&\le \quad (C_b - C_{b-1}) &\forall v  \in V, \forall b \in [B]\\
\sum_{e \in E[U]}x_{e,b}\quad&\le\quad (C_b - C_{b-1}) \cdot \lfloor \tfrac{1}{2}|U|\rfloor &\quad\forall U \subseteq V : |U| \text{ odd}, \forall b \in [B] \\
x_{e,b} \quad&= \quad 0 &\forall e \in E, \forall b > b(e) \\
x_{e,b} \quad&\ge\quad 0
\end{aligned}
\end{gather*}
The variables $x_{e,b}$ represent whether edge $e \in E$ is assigned to block $b \in [B]$. Here, $E$ and $V$ are the edge and respectively vertex sets of the underlying graph. For an edge $e \in E$, the function $b(e)$ returns the smallest index $i \in [B]$ of the jobs to which this edge belongs. Before explaining the structure, we slightly modify the LP. For some fixed $\tau  \in \mathbb{N}_{\ge 2}$, we round the deadlines $C_1,\dots,C_B$ up to the next multiple of $\tau$. Let the resulting deadlines be $C_1',\dots, C_B'$ and for all $b \in [B]$ define $k_b \in \mathbb{N}$ such that $C'_b - C'_{b-1} = \tau \cdot k_b$. Using these definitions, we obtain the following LP:

\begin{gather*}\tag{\textbf{LP}}\label{lp:main_gg}
\begin{aligned}
\sum_{b \in [B]}x_{e,b} \quad &= \quad 1 &\forall e \in E\\\sum_{e: v \in e} x_{e,b} \quad&\le \quad k_b \cdot \tau \ &\forall v \in V,\forall b \in [B]\\ \sum_{e \in E[U]}x_{e,b}\quad&\le\quad k_b \cdot \tau \cdot \lfloor \tfrac{1}{2}|U|\rfloor &\quad\forall U \subseteq V : |U| \text{ odd}, \forall b \in [B] \\x_{e,b} \quad&= \quad 0 &\forall e \in E, \forall b > b(e) \\ x_{e,b} \quad&\ge\quad 0
\end{aligned}
\end{gather*}
As explained earlier, we group the previously unit sized time slots into blocks between consecutive deadlines. The first set of constraints ensures that each edge is assigned to some block, while the second set ensures that in each block at most the block's size many edges are adjacent to each vertex. Additionally to these constraints, we also have block versions of the so called odd-set constraints, which ensure that in each block $b \in [B]$ in every odd sized vertex set $|U| \subseteq V$, at most $k_b \cdot \tau \cdot \lfloor \tfrac{1}{2}|U|\rfloor$ many edges can be included. These types of inequalities are well known and were first introduced by Edmonds \cite{edmonds1965maximum}. In the case of the standard matching polytope, for non bipartite graphs these constraints are necessary to ensure that all polytope vertices are integral. In our case, vertex solutions in general will not be integral, but the added structure will be exploited in \cref{sec:int_time_slot_ass} to ensure that certain good schedules exist and can be efficiently constructed.

Note that for \cref{lp:main_gg} without the odd-set constraints there exists an algorithm determining deadlines such that the LP is feasible and such that the weighted sum over these deadlines is $2$-approximate with respect to the optimal cost of the underlying Graph Scheduling instance. This algorithm has been used in previous works for Coflow Scheduling and related problems \cite{im19,fukunaga22}. We show that the algorithm can be extended to our setting with the same guarantees, see \cref{sec:app:2appround} for details.

\subsection{Iterated LP Rounding}\label{sec:iter_lp_round}
In this section we describe how we round a fractional solution to \cref{lp:main_gg} to obtain an integer point, which can be turned into a valid schedule using a procedure which is explained in \cref{sec:int_time_slot_ass}. This integer point might violate some of the original constraints of \cref{lp:main_gg}, but we show that we can control the magnitude of this violation and thus obtain strong delay guarantees for the resulting schedule.

Let $k \in \mathbb{N}$ be a fixed constant. Our rounding procedure has two steps. First we solve \cref{lp:main_gg} and fix every variable which is integral in the solution. In the second step, we perform two modifications to the constraint sets. We drop every vertex $v$ which now has strictly less than $k$ fractional edges adjacent to it. Dropping a vertex means removing the corresponding vertex degree constraint. For the odd-set constraints, whenever we drop some vertex, we stop considering its adjacent fractional edges in any odd-set constraint. So it is possible that some edges only contribute to degree constraints of one vertex but are not considered for the odd-set inequalities.

Let $E_b^{(f)}$ be the set of fractional edges remaining in block $b \in [B]$ and let $V^{(f)}_b$ be the set of vertices in that block with at least $k$ fractional edges adjacent. Let $E^{(1)}$ be the set of edges which have been fixed to $1$ in a previous step and $E^{(0)}$ respectively the ones fixed to $0$. Let $E^{(d)}_b \subseteq E$ be the set of dropped edges in block $b$, meaning that they were fractional at the time one of their vertices was dropped. Using these definitions, in any step after the first one we have the following LP:
\begin{gather*}\tag{\textbf{LP R}} \label{lp:lp_main_red}
\begin{aligned}
\sum_{b \in [B]}x_{e,b} \quad &= \quad 1 &\forall e \in \bigcup_{b \in [B]}E_b^{(f)}\\
\sum_{e: v \in e} x_{e,b} \quad&\le \quad k_b \cdot \tau &b \in [B],\ \forall v \in V^{(f)}_b\\
\sum_{e \in E[U] \setminus E^{(d)}_b}x_{e,b}\quad&\le\quad k_b \cdot \tau \cdot \lfloor \tfrac{1}{2}|U|\rfloor &\quad\forall U \subseteq V: |U| \text{ odd}, \forall b \in [B] \\
x_{e,b} \quad&= \quad 0 &\forall e \in E^{(0)} \\
x_{e,b} \quad&= \quad 1 &\forall e \in E^{(1)}\\
x_{e,b} \quad&\ge\quad 0
\end{aligned}
\end{gather*}
Note that the odd-set inequalities can still be separated in polynomial time by considering the subgraph of only the edges in $E \setminus E_b^{(d)}$ and using the standard separation procedure \cite{padberg82_separation,Schrijver2002_combopt}. We repeat the fixing variables, dropping constraints, and re-solving of the resulting LP until all variables are integral. Whenever we solve the LP, we assume that we obtain a basic feasible solution, which can be done in polynomial time \cite{Tardos1986}.
Analyzing this procedure requires two major ingredients. We have to bound the maximal constraint violations which the integral point will have with respect to \cref{lp:main_gg} and we have to show that the procedure is guaranteed to terminate and run in polynomial time.

\subsection{Bounding Violations}
Assuming that the procedure described in the previous section terminates, we obtain an integral point, which potentially violates some of the constraints of \cref{lp:main_gg}. We give exact upper bounds on this violation:
\begin{lemma}\label{lemma:lp_sol_violations}
For the integral point obtained through the rounding procedure it holds:
\begin{itemize}
    \item[a)] The vertex degree constraints in \cref{lp:main_gg} are exceeded by at most $k-2$.
    \item[b)] For each $b \in [B]$ and $ U \subseteq V_b^{(f)}$, the corresponding odd-set constraint in \cref{lp:main_gg} is violated by at most $2(k-1) \cdot\lfloor \tfrac{1}{2}|U|\rfloor-1$.
\end{itemize}
\end{lemma}
\begin{proof}
Consider some fixed vertex $v \in V$ and the corresponding vertex constraint $\sum_{e:v\in e} x_e \le k_b \cdot \tau$. This constraint gets dropped when at most $k-1$ edges adjacent to $v$ are still fractional. Integral and thus fixed edges never get changed again, so only these $k-1$ edges could change compared to the solution at the moment the constraint is dropped. If they all get included and thus set to $1$, the total sum on the left side would increase by at most $k-1$. However, as these edges are fractional at the point the constraint gets dropped, the sum over the fixed edges at that point can be at most $k_b \cdot \tau - 1$, hence the maximum violation of vertex constraints is bounded by $k-2$.
For the odd-set constraints, consider some vertex set $U \subseteq V$ of odd size. The corresponding odd-set constraint never gets dropped from the LP during the procedure, however the considered set of edges will change over time. Every time a vertex $v$ from $U$ is dropped, the fractional edges adjacent to $v$ in $E[U]$ will no longer be considered in this constraint in future iterations. Consider the vertices $v_1,\dots,v_{|U|}$ in $U$ by the order in which they are dropped, with all non-dropped vertices at the end. By definition, there are at most $k-1$ fractional edges adjacent to a vertex at the time it is dropped. However, when counting the edges we only need to consider the ones which are connected to vertices later in the order, as the ones connected to earlier vertices will have been counted in an earlier step already. Therefore, specifically the very last vertex can never contribute to the edge count in this counting argument. Thus there are at most $(|U|-1)\cdot(k-1)$ fractional edges contained in $E_b^{(d)}[U]$. Through the same fractionality argument as above, this bound can further be strengthened by $-1$, as the sum over the fixed edges has to be strictly smaller than the upper bound to allow fractional edges to exist. So, assuming that all fractional dropped edges later get included, the constraint bound is exceeded by at most $(|U|-1)\cdot (k-1) -1 = 2(k-1) \cdot \lfloor \tfrac{1}{2}|U|\rfloor -1$.
\end{proof}

\subsection{Counting Constraints}
In order to establish that the procedure terminates, we show that in each step we make progress, meaning that we obtain at least one new integral variable. To prove this, we bound the number of tight linearly independent constraints and show that this is strictly less than the number of variables, which by standard arguments from LP theory suffices to establish the claim. When counting tight constraints in \cref{lp:lp_main_red}, we do not consider the constraints of the form $x_{e,b} = 0$, $x_{e,b} = 1$ or $x_{e,b} \ge 0$, as them being tight directly implies existence of an integral variable. Also the fixing constraints and their respective variables are in a one to one correspondence, so these variables do not affect the inequality argument and are thus omitted from the count.

The total number of variables in each intermediate step is $\sum_{b \in [B]}|E_b^{(f)}|$. We independently for each block establish certain bounds on the number of tight linearly independent constraints. Let $b \in [B]$ be thus fixed and as previously consider the sets $V^{(f)}_b$ and $E_b^{(f)}$. For each $v \in V_b^{(f)}$ and for each $e \in \bigcup_{b \in [B]}E_b^{(f)}$ we obtain one constraint.

\begin{lemma}\label{lemma:frac_v_bound}
    For all $b \in [B]:$ $\quad|V_b^{(f)}| \qle \frac{2}{k}|E_b^{(f)}|$
\end{lemma}
\begin{proof}
By definition, every vertex in $V_b^{(f)}$ has at least $k$ fractional edges adjacent, so for every $b \in [B]$ we have
\begin{equation*}
    k \cdot |V_b^{(f)}| \hle |\{(v,e)\ |\ v \in V_b^{(f)}, e \in E_b^{(f)} \cap \delta(v) \}|.
\end{equation*}
Every edge contains exactly two vertices, so we also have the bound
\begin{equation*}
    |\{(v,e)\ |\ v \in V_b^{(f)}, e \in E_b^{(f)} \cap \delta(v) \}| \hle 2 \cdot |E_b^{(f)}|.
\end{equation*}
Combining the two and rearranging yields the result.
\end{proof}

\begin{lemma}\label{lemma:frac_e_bound}
It holds $\quad|\bigcup_{b \in [B]}E_b^{(f)}| \qle \frac{1}{2}\sum_{b \in [B]}|E_b^{(f)}|$.
\end{lemma}
\begin{proof}
The first set of constraints in \cref{lp:lp_main_red} prescribes that $\sum_{b \in [B]}x_{e,b} = 1$, which implies that whenever some edge $e$ has a fractional assignment in some block $b$, there has to exist at least one other block in which this edge is also fractionally assigned. Hence such an edge only contributes once to the left side of the inequality and contributes at least twice to the right side, yielding the claimed bound.
\end{proof}

\cref{lemma:frac_v_bound} and \cref{lemma:frac_e_bound} provide bounds for the first two sets of constraints, so it remains to establish a similar bound for the odd-set inequalities. A priori the number of inequalities is exponential, as there is one inequality for each subset $U \subseteq V$ of odd size. However, we claim that the number of such tight constraints which are linearly independent of each other and all other constraints in \cref{lp:lp_main_red} is actually linear in $|V_b^{(f)}|$. Whenever we write linearly independent we refer to this notion of independence.

\begin{lemma}\label{lemma:odd_set_bound}
    In \cref{lp:lp_main_red}, the number of tight linearly independent odd-set constraints $|\mathcal{O}_b|$ in block $b \in [B]$ is upper bounded by $|V_b^{(f)}| - 1$.
\end{lemma}
\begin{proof}
We establish the result in two steps. We first show certain structure which linearly independent constraints must obey and then establish the bound using an inductive argument. 

It is well known that the vertex sets belonging to the tight odd-set constraints form a laminar family $\mathcal{L}$, which can be shown using standard uncrossing arguments \cite{schrijver1977laminar}. We claim that there cannot be an extension of a laminar set using only dropped vertices, meaning that if there is a set $L \in \mathcal{L}$ and set $L_1\in \mathcal{L}$ which is subset of $L$ such that $L \setminus L_1$ only consists of dropped vertices, then $L$ is linearly dependent of $L_1$ and possibly some tight edge constraints. As all vertices in $L \setminus L_1$ were dropped at some point, by definition no fractional edge adjacent to them will be considered in the constraint belonging to $L$. So the only possible difference to the constraint belonging to $L_1$ is in variables belonging to fixed edges adjacent to vertices in $L \setminus L_1$. However, for each of these by definition there is a equality constraint setting them to $1$. Therefore, the constraint of $L$ can be written as a linear combination of these fixed edge constraints and the constraint of $L_1$, showing that it is not linearly independent.

We show using induction that the maximum cardinality laminar family representing tight linearly independent odd-sets on $n \ge 2$ non-dropped vertices (and arbitrarily many dropped vertices) has size at most $n-1$. Any odd-set which induces a constraint that is not simply a linear combination of fixed edge constraints must contain at least $2$ non-dropped vertices, as otherwise it cannot contain any fractional edges. For $n=2$, this leaves a single set with $2$ non-dropped vertices and some odd number of dropped vertices as the only option, as any other set would either be disjoint and thus contain only dropped vertices or it would have to be an extension using only dropped vertices, which by the previous argument is not possible.

Now consider some laminar family $\mathcal{L}$ on $n \ge 3$ non-dropped vertices. If $\mathcal{L}$ contains more than one set $L_1,L_2,\dots,L_i$, each not contained in any larger set, then each has to contain at least two non-dropped vertices and thus also strictly fewer than $n$ non-dropped vertices. Thus by the induction hypothesis, if $\ell_1,\dots,\ell_i$ are the respective number of non-dropped vertices in the sets, each can consist of at most $\ell_1-1,\dots,\ell_i-1$ sets. Thus the total number of sets is at most $n - i < n-1$. If no two or more such sets exist, then either $\mathcal{L}$ contains no set, in which case the claim is trivially true, or it contains exactly one such maximal set $L$. If $L$ contains exactly one maximal set $L_1$, then it has to contain at least one non-dropped vertex not in $L_1$ to be linearly independent, so applying the induction hypothesis to $L_1$ directly yields the result. If $L$ contains $i \ge 2$ maximal sets, then again each has to contain non-dropped vertices, so we can apply the induction hypothesis to each and get an upper bound of $n-i$ contained sets in $L$. Together with $L$ itself this yields the bound $n-1$.

Note that a single chain of $n-1$ nested sets, where the innermost set contains $2$ non-dropped and an odd number of dropped vertices and every other set extends its predecessor by a single non-dropped and an odd number of dropped vertices has the claimed size, so this bound is tight.
\end{proof}

\subsection{Integral Solution Guarantees}
Using the just established results, we derive the following lemma, which gives precise bounds on the achieved constraint violations for the integral point obtained through the iterated rounding procedure.

\begin{lemma}\label{lemma:iter_round_violation_bound}
Given integral deadlines $C_1,\dots,C_B$ for which \cref{lp:main_gg} is feasible, in polynomial time we can find an integral point such that:
\begin{itemize}
    \item[a)] The vertex degree constraints in \cref{lp:main_gg} are exceeded by at most $6$.
    \item[b)] For each $b\in [B]$ and $U \subseteq V$, the corresponding odd-set inequality in \cref{lp:main_gg} is exceeded by at most $14 \lfloor\tfrac{1}{2}|U|\rfloor -1$.
    \item[c)] All other constraints in \cref{lp:main_gg} are fulfilled.
\end{itemize}
\end{lemma}
\begin{proof}
We have to show two things, namely that the iterated rounding procedure terminates in polynomial time and that all of the claimed constraint (non-) violations hold. For the runtime, we show that in each iteration the number of linearly independent constraints is strictly lower than the number of variables, which by standard LP theory implies that in a basic solution at least one variable must be integral. Combining the bounds from \cref{lemma:frac_v_bound}, \cref{lemma:frac_e_bound}, and \cref{lemma:odd_set_bound} we obtain for the constraints $\algo{Cons}$ and variables $\algo{Vars}$ in \cref{lp:lp_main_red}:
\begin{align*}
    |\algo{Cons}| &\qeq |\bigcup_{b \in [B]} E_b^{(f)}| \hop{+} \sum_{b \in [B]}|V_b^{(f)}| \hop{+} \sum_{b \in [B]} |\mathcal{O}_b|\\~&\qle \mfrac{1}{2}\sum_{b \in [B]}|E_b^{(f)}| \hop{+} \mfrac{2}{k}\sum_{b \in [B]}|E_b^{(f)}| \hop{+} \mfrac{2}{k}\sum_{b \in [B]}|E_b^{(f)}| \hop{-} B\\~&\qeq \Big(\mfrac{1}{2} + \mfrac{4}{k}\Big)|\algo{Vars}| \hop{-} B
\end{align*}
So for $k \ge 8$ we have the desired strict inequality, implying that there is at least one new fixed integral variable in every iteration. From this it directly follows that the procedure runs in polynomial time, as only a polynomial number of iterations of the rounding procedure are needed and as a basic solution to \cref{lp:lp_main_red} can be determined in polynomial time.
It is easy to see that c) must hold, as these other constraints are never modified and thus remain valid during each step of the rounding procedure. For the choice $k=8$, the claims in a) and b) directly follow from the violation bounds shown in \cref{lemma:lp_sol_violations}.
\end{proof}

\section{Integral Time Slot Assignments}\label{sec:int_time_slot_ass}
Given an integral assignment of edges to blocks, we want to determine a corresponding assignment of these edges to individual time slots. For bipartite graphs, Kőnig's Theorem provides a simple tool to obtain such assignments, but as discussed in \cref{sec:intro:ggextension}, for general multigraphs no similarly strong guarantees hold. However, using further graph properties apart from the maximum degree, useful guarantees can still be established. Further information and historic context for the concepts used and for general edge coloring can be found for example in \cite{scheide2009_edgecolouring,Cao2019-ow,toft2021_edgecoloring}.

Due to space constraints we refer for most of the details to \cref{sec:app:graph_theory} in the appendix and only briefly recap the main ideas here. Using the odd-set constraints in \cref{lp:main_gg} and the guarantees from \cref{lemma:iter_round_violation_bound}, we show that the so called density of each graph belonging to each block $b \in [B]$ in an integral solution obtained by the rounding procedure is at most $k_b \cdot \tau + 14$. Using recent results from edge-coloring theory for multigraphs \cite{jing2024edgecoloringmultigraphs}, this enables us to schedule all edges in each block in at most $k_b \cdot \tau + 14$ time slots. From this result, a procedure with the following guarantees follows, where $\algo{GBF}^\tau(C_j)$ refers to the finishing time of job $j$ in the returned schedule.

\begin{restatable}{lemma}{algdelaybound}\label{lemma:alg_delay_bound}
For given deadlines $C_1,\dots, C_B$ for which \ref{lp:main_gg} is feasible and a parameter $\tau \in \mathbb{N}_{\ge 2}$, there is an algorithm $\algo{GBF}^\tau$ returning a valid schedule such that the following holds for all $j \in [B]$:
\begin{equation*}
    \algo{GBF}^\tau(C_j) \qle \mfrac{\tau+ 14}{\tau}C_j \hop{+} \tau \hop{+} 14
\end{equation*}
\end{restatable}

\subsection{Scheduling Algorithm}\label{sec:scheduling_algorithm}
The guarantees from \cref{lemma:iter_round_violation_bound} provide bounds on the delay incurred through the iterated rounding procedure, however the lemma implicitly relies on the deadlines used in the definition of \cref{lp:main_gg}. By combining this lemma with an approximation algorithm which generates deadlines, we obtain a more explicit parameterized scheduling algorithm.

\mainschedulingalg*

\begin{proof}
Let $G=(V,E)$ be a multigraph representing a Graph Scheduling instance, together with jobs $J_1,\dots,J_B \subseteq E$ and a weight vector $\omega \in \mathbb{R}^B_{\ge 0}$.  From \cref{lemma:app:two_approx} shown in the Appendix in \cref{sec:app:2appround}, we obtain that there is an algorithm returning deadlines for this instance for which \cref{lp:main_gg} is feasible such that the following upper bound holds with respect to the optimal cost $\algo{OPT}$:
\begin{equation*}
\sum_{j \in [B]} \omega_j C_j \qle  2\cdot \algo{OPT} \hop{-}  \sum_{j \in [B]}\omega_j
\end{equation*}
Combining these deadlines with the procedure from \cref{lemma:alg_delay_bound} yields the following guarantees for any $\tau \ge 2$. Here again $\algo{GBF}^\tau(C_j)$ is the finishing time of job $J_j$ in the schedule returned by the procedure.\hypertarget{link:main_scheduling_alg_proof}{}
\begin{align*}
    \sum_{j \in [B]}\omega_j \algo{GBF}^\tau(C_j) &\qle \sum_{j \in [B]}\omega_j\Big(\mfrac{\tau+14}{\tau}C_j + \tau +14\Big) \\
    &\qeq \mfrac{\tau+14}{\tau}\sum_{j \in [B]}\omega_jC_j \hop{+} (\tau + 14) \sum_{j \in [B]}\omega_j \\ ~ &\qle \mfrac{\tau + 14}{\tau}\Big(2\cdot \algo{OPT} - \sum_{j \in [B]} \omega_j\Big) \hop{+} (\tau + 14)\sum_{j \in [B]}\omega_j \\&\qeq 2 \cdot \mfrac{\tau+14}{\tau} \cdot \algo{OPT} \hop{+} \big(\tau+14-\tfrac{\tau+14}{\tau}\big)\sum_{j \in [B]}\omega_j
\end{align*}
As $\tau$ can be chosen arbitrarily large, for any $\epsilon > 0$ there is a $\tau$ such that we obtain a cost bound of 
\begin{equation*}
    (2+\epsilon)\cdot \algo{OPT} \hop{+} \mathcal{O}\Big(\mfrac{1}{\epsilon}\Big)\cdot \sum_{j \in [B]}\omega_j. \qedhere
\end{equation*}
\end{proof}
Note that we are focusing on the technical ideas in this work and have not fully optimized the various constants appearing throughout. The same ideas as used in \cite{rohwedder2025_coflow} should provide small improvements for the guarantees shown in \cref{thm:main_scheduling_alg} and further enhancements are likely possible.

\section{Data Migration}\label{sec:data_migration}
Data Migration is defined as Graph Scheduling on a general multigraph $G=(V,E)$ with a special job structure. For each vertex $v \in V$ there is a job $J_v := \{e \in \delta_E(v)\}$, so $J_v$ contains all edges adjacent to $v$. This is equivalent to considering each vertex as a job and defining a vertex to finish whenever all its adjacent edges have been scheduled. As the name suggests, the setting is inspired by a problem in distributed computation and storage, where the vertices for example represent some distributed servers and each edge represents some data migration request. We only consider unit sized edges here as the setting with general demands requires very different techniques.

\subsection{$(1 + \phi)$-Approximation}
The so far best known approximation algorithm for Data Migration provides a $(1+\phi) \approx 2.618$-approximation \cite{mestre2010}. Here $\phi \approx 1.618$ represents the golden ratio. The algorithm is based on a technique called local ratio, where guarantees are first shown on a local scale and then lifted to a general bound. Through more fine-grained analysis of the proofs in \cite{mestre2010}, we are able to show a slight additive improvement over their guarantees.
\begin{restatable}{lemma}{localratioimprov}\label{lemma:local_ratio_bound}
For any instance of Data Migration with cost vector $\omega \in \mathbb{R}_{\ge 0}^B$, there is an algorithm returning a solution $S$ such that:
\begin{equation*}
    \omega(S) \qle (1+\phi)\cdot\algo{OPT} \hop{-} \mfrac{1}{2}\sum_{j \in [B]}\omega_j
\end{equation*}
\end{restatable}
Here $\algo{OPT}$ refers to the optimal cost of the underlying problem instance. For a full proof of this lemma and more details about the algorithm and the involved techniques, see \cref{sec:app:datamigration}.

\subsection{Combined Approximation}

Using our results for general Graph Scheduling combined with the improved analysis of the local ratio procedure for Data Migration, we show a slightly improved approximation guarantee. As a first step, using the results from \cref{sec:scheduling_algorithm}, we obtain the following bound.

\begin{lemma}\label{lemma:dm_tau_choice}
For any instance of Data Migration with cost vector $\omega \in \mathbb{R}_{\ge 0}^B$, there is an algorithm returning a solution $S$ such that:
\begin{equation*}
    \omega(S) \qle \Big(2 + \mfrac{28}{97}\Big)\cdot\algo{OPT} \hop{+} \mfrac{10656}{97}\sum_{j \in [B]}\omega_j
\end{equation*}
\begin{proof}
Follows directly from the bound in the \hyperlink{link:main_scheduling_alg_proof}{proof of \cref*{thm:main_scheduling_alg}} with $\tau = 97$.
\end{proof}
\end{lemma}
Combining the guarantees from \cref{lemma:local_ratio_bound} and \cref{lemma:dm_tau_choice} leads to the following purely multiplicative bound, which is an improvement over the best known $(1+\phi)$-approximation.

\dmapproxresult*

\begin{proof}
The procedure is straightforward: Run both the algorithms from \cref{lemma:local_ratio_bound} and \cref{lemma:dm_tau_choice} independently and return the schedule with lower cost. Let $S_1$ and $S_2$ be the respective schedules returned by the algorithm. The overall cost is then given by $\min\{\omega(S_1),\omega(S_2)\}$.

If instead of returning the schedule with lower cost, we return the first schedule with some probability $\lambda \in [0,1]$ and the second with probability $1-\lambda$, we obtain a random schedule $S_R$ with cost at least as high.
Bounding the costs, we obtain:
\begin{align*}
    \min\{\omega(S_1),\omega(S_2)\} &\qle \mathbb{E}[\omega(S_R)] \\
    ~&\qle \lambda \omega(S_1) \hop{+} (1-\lambda)\omega(S_2)\\
    ~&\qeq \left(\lambda(1+\phi) + (1-\lambda)\mfrac{222}{97}\right) \algo{OPT} \hop{+} \left(-\frac{\lambda}{2} + (1-\lambda)\mfrac{10656}{97}\right)\sum_{j \in [B]}\omega_j
\end{align*}
For the choice $\lambda = \frac{21312}{21409}$, the additive term vanishes and the first term gives the desired multiplicative factor $2.6165\dots < 2.617$
\end{proof}

\bibliographystyle{plainurl}
\bibliography{references}

@inproceedings{agarwal_sincronia_2018,
	title        = {Sincronia: near-optimal network design for coflows},
	shorttitle   = {Sincronia},
	author       = {Agarwal, Saksham and Rajakrishnan, Shijin and Narayan, Akshay and Agarwal, Rachit and Shmoys, David and Vahdat, Amin},
	year         = {2018},
	booktitle    = {Proceedings of SIGCOMM},
	publisher    = {ACM},
	pages        = {16--29},
	doi          = {10.1145/3230543.3230569},
	isbn         = {978-1-4503-5567-4}
}

@incollection{ahmadi_scheduling_2017,
	title        = {On {Scheduling} {Coflows}},
	shorttitle   = {On {Scheduling} {Coflows}},
	author       = {Ahmadi, Saba and Khuller, Samir and Purohit, Manish and Yang, Sheng},
	year         = {2017},
	booktitle    = {Proceedings of IPCO},
	publisher    = {Springer International Publishing},
	volume       = {10328},
	pages        = {13--24},
	doi          = {10.1007/978-3-319-59250-3_2},
	isbn         = {978-3-319-59250-3},
	note         = {Series Title: Lecture Notes in Computer Science}
}

@incollection{Bar-Yehuda1985-se,
	title        = {A local-ratio theorem for approximating the weighted vertex cover problem},
	author       = {Bar-Yehuda, Reuven and Even, Shimon},
	year         = {1985},
	booktitle    = {Analysis and Design of Algorithms for Combinatorial Problems},
	publisher    = {Elsevier},
	series       = {North-Holland mathematics studies},
	pages        = {27--45},
	doi          = {10.1016/S0304-0208(08)73101-3}
}

@article{Bar-Yehuda2004-vk,
	title        = {Local ratio},
	author       = {Bar-Yehuda, Reuven and Bendel, Keren and Freund, Ari and Rawitz, Dror},
	year         = {2004},
	journal      = {ACM Comput. Surv.},
	publisher    = {Association for Computing Machinery (ACM)},
	volume       = {36},
	number       = {4},
	pages        = {422--463},
	doi          = {10.1145/1041680.1041683}
}

@article{Cao2019-ow,
	title        = {Graph edge coloring: A survey},
	author       = {Cao, Yan and Chen, Guantao and Jing, Guangming and Stiebitz, Michael and Toft, Bjarne},
	year         = {2019},
	journal      = {Graphs Comb.},
	publisher    = {Springer Science and Business Media LLC},
	volume       = {35},
	number       = {1},
	pages        = {33--66},
	doi          = {10.1007/s00373-018-1986-5}
}

@article{Chen2025_goldberg,
	title        = {Proof of the {Goldberg--Seymour} conjecture on edge--colorings of multigraphs},
	author       = {Chen, Guantao and Jing, Guangming and Zang, Wenan},
	year         = {2025},
	journal      = {J. Comb. Optim.},
	publisher    = {Springer Science and Business Media LLC},
	volume       = {50},
	number       = {3},
	doi          = {10.1007/s10878-025-01348-6},
	copyright    = {https://creativecommons.org/licenses/by/4.0}
}

@inproceedings{chowdhury_efficient_2014,
	title        = {Efficient coflow scheduling with {Varys}},
	author       = {Chowdhury, Mosharaf and Zhong, Yuan and Stoica, Ion},
	year         = {2014},
	booktitle    = {Proceedings of {SIGCOMM}},
	publisher    = {ACM},
	pages        = {443--454},
	doi          = {10.1145/2619239.2626315},
	isbn         = {978-1-4503-2836-4}
}

@article{Dean2008,
	title        = {MapReduce: simplified data processing on large clusters},
	author       = {Dean,  Jeffrey and Ghemawat,  Sanjay},
	year         = {2008},
	journal      = {Communications of the ACM},
	publisher    = {Association for Computing Machinery (ACM)},
	volume       = {51},
	number       = {1},
	pages        = {107–113},
	doi          = {10.1145/1327452.1327492},
	issn         = {1557-7317}
}

@article{edmonds1965maximum,
	title        = {Maximum matching and a polyhedron with 0, 1-vertices},
	author       = {Edmonds, Jack},
	year         = {1965},
	journal      = {Journal of Research of the National Bureau of Standards B},
	volume       = {69},
	pages        = {125--130}
}

@inproceedings{fukunaga22,
	title        = {{Integrality Gap of Time-Indexed Linear Programming Relaxation for Coflow Scheduling}},
	author       = {Fukunaga, Takuro},
	year         = {2022},
	booktitle    = {Proceedings of APPROX/RANDOM},
	publisher    = {Schloss Dagstuhl -- Leibniz-Zentrum f{\"u}r Informatik},
	address      = {Dagstuhl, Germany},
	series       = {LIPIcs},
	volume       = {245},
	pages        = {36:1--36:13},
	doi          = {10.4230/LIPIcs.APPROX/RANDOM.2022.36},
	isbn         = {978-3-95977-249-5},
	issn         = {1868-8969},
	urn          = {urn:nbn:de:0030-drops-171581}
}

@article{Gandhi2009-edgecolor,
	title        = {Combinatorial algorithms for data migration to minimize average completion time},
	author       = {Gandhi, Rajiv and Mestre, Juli{\'a}n},
	year         = {2009},
	journal      = {Algorithmica},
	publisher    = {Springer Science and Business Media LLC},
	volume       = {54},
	number       = {1},
	pages        = {54--71},
    doi          = {10.1007/s00453-007-9118-2},
	copyright    = {https://creativecommons.org/licenses/by-nc/2.0}
}

@article{Gandhi2013_datamigration,
	title        = {Corrigendum: Improved Results for Data Migration and Open Shop Scheduling},
	author       = {Gandhi, Rajiv and Halld{\'o}rsson, Magn{\'u}s M and Kortsarz, Guy and Shachnai, Hadas},
	year         = {2013},
	journal      = {ACM Trans. Algorithms},
	publisher    = {Association for Computing Machinery (ACM)},
	volume       = {9},
	number       = {4},
	pages        = {1--7},
	doi          = {10.1145/2500123}
}

@inproceedings{Hadoop2010,
	title        = {The Hadoop Distributed File System},
	author       = {Shvachko, Konstantin and Kuang, Hairong and Radia, Sanjay and Chansler, Robert},
	year         = {2010},
	booktitle    = {Proceedings of IEEE MSST},
	pages        = {1--10},
	doi          = {10.1109/MSST.2010.5496972}
}

@article{Halldorsson2003-edgecolor,
	title        = {Sum coloring interval and k-claw free graphs with application to scheduling dependent jobs},
	author       = {Halld{\'o}rsson, Magn{\'u}s M and Kortsarz, Guy and Shachnai, Hadas},
	year         = {2003},
	journal      = {Algorithmica},
	publisher    = {Springer Science and Business Media LLC},
	volume       = {37},
	number       = {3},
	pages        = {187--209},
    doi          = {10.1007/s00453-003-1031-8}
}

@incollection{Halldorsson2008-msec,
	title        = {Min sum edge coloring in multigraphs via configuration {LP}},
	author       = {Halld{\'o}rsson, Magn{\'u}s M and Kortsarz, Guy and Sviridenko, Maxim},
	year         = {2008},
	booktitle    = {Integer Programming and Combinatorial Optimization},
	publisher    = {Springer Berlin Heidelberg},
	pages        = {359--373},
    doi          = {10.1007/978-3-540-68891-4_25}
}

@article{holyer1981_edgecolorhard,
	title        = {The NP-Completeness of Edge-Coloring},
	author       = {Holyer, Ian},
	year         = {1981},
	journal      = {SIAM Journal on Computing},
	volume       = {10},
	number       = {4},
	pages        = {718--720},
	doi          = {10.1137/0210055}
}

@inproceedings{im19,
	title        = {Matroid {Coflow} {Scheduling}},
	author       = {Im, Sungjin and Moseley, Benjamin and Pruhs, Kirk and Purohit, Manish},
	year         = {2019},
	booktitle    = {Proceedings of ICALP},
	publisher    = {Schloss Dagstuhl–Leibniz-Zentrum fuer Informatik},
	series       = {{LIPIcs}},
	volume       = {132},
	pages        = {145:1--145:14},
	doi          = {10.4230/LIPIcs.ICALP.2019.145},
	isbn         = {978-3-95977-109-2},
	note         = {ISSN: 1868-8969}
}

@misc{jing2024edgecoloringmultigraphs,
	title        = {On Edge Coloring of Multigraphs},
	author       = {Guangming Jing},
	year         = {2024},
	eprint       = {2308.15588},
	archiveprefix = {arXiv},
	primaryclass = {math.CO}
}

@article{khuller_select_2019,
	title        = {Select and permute: {An} improved online framework for scheduling to minimize weighted completion time},
	author       = {Khuller, Samir and Li, Jingling and Sturmfels, Pascal and Sun, Kevin and Venkat, Prayaag},
	year         = {2019},
	journal      = {Theoretical Computer Science},
	volume       = {795},
	pages        = {420--431},
	doi          = {10.1016/j.tcs.2019.07.026},
	note         = {Publisher: Elsevier BV}
}

@article{König1916,
	title        = {Über Graphen und ihre Anwendung auf Determinantentheorie und Mengenlehre},
	author       = {König, Dénes},
	year         = {1916},
	journal      = {Mathematische Annalen},
	volume       = {77},
	pages        = {453--465},
	doi          = {10.1007/BF01456961}
}

@misc{megow2025groupschedule,
	title        = {Unifying Scheduling Algorithms for Group Completion Time},
	author       = {Alexander Lindermayr and Zhenwei Liu and Nicole Megow},
	year         = {2025},
	eprint       = {2501.17682},
	archiveprefix = {arXiv},
	primaryclass = {cs.DS}
}

@article{mestre2010,
	title        = {Adaptive Local Ratio},
	author       = {Mestre, Juli{\'a}n},
	year         = {2010},
	journal      = {SIAM J. Comput.},
	publisher    = {Society for Industrial \& Applied Mathematics (SIAM)},
	volume       = {39},
	number       = {7},
	pages        = {3038--3057},
	doi          = {10.1137/080731712}
}

@article{padberg82_separation,
	title        = {Odd Minimum Cut-Sets and b-Matchings},
	author       = {Manfred W Padberg and Mendu R Rao},
	year         = {1982},
	journal      = {Mathematics of Operations Research},
	publisher    = {INFORMS},
	volume       = {7},
	number       = {1},
	pages        = {67--80},
	issn         = {0364765X, 15265471}
}

@inproceedings{rohwedder2025_coflow,
	title        = {{3.415-Approximation for Coflow Scheduling via Iterated Rounding}},
	author       = {Rohwedder, Lars and Schnaars, Leander},
	year         = {2025},
	booktitle    = {Proceedings of ICALP},
	publisher    = {Schloss Dagstuhl–Leibniz-Zentrum fuer Informatik},
	series       = {LIPIcs},
	volume       = {334},
	pages        = {128:1--128:19},
	doi          = {10.4230/LIPIcs.ICALP.2025.128},
	isbn         = {978-3-95977-372-0},
	issn         = {1868-8969},
	urn          = {urn:nbn:de:0030-drops-235050}
}

@inproceedings{Sachdeva2013,
	title        = {Optimal Inapproximability for Scheduling Problems via Structural Hardness for Hypergraph Vertex Cover},
	author       = {Sachdeva,  Sushant and Saket,  Rishi},
	year         = {2013},
	booktitle    = {2013 IEEE Conference on Computational Complexity},
	publisher    = {IEEE},
	pages        = {219–229},
	doi          = {10.1109/ccc.2013.30}
}

@phdthesis{scheide2009_edgecolouring,
	title        = {Edge Colourings of Multigraphs},
	author       = {Scheide, Diego},
	year         = {2009},
	day          = {24},
	school       = {Technische Universität Ilmenau}
}

@book{schrijver1977laminar,
	title        = {A proof of total dual integrality of matching polyhedra},
	author       = {Schrijver, Alexander and Seymour, Paul D},
	year         = {1977},
	publisher    = {Stichting Mathematisch Centrum. Zuivere Wiskunde}
}

@book{Schrijver2002_combopt,
	title        = {Combinatorial Optimization},
	author       = {Schrijver, Alexander},
	year         = {2002},
	publisher    = {Springer},
	address      = {Berlin, Germany},
	series       = {Algorithms and Combinatorics},
	edition      = {2003}
}

@article{Shafiee2018-ik,
	title        = {An improved bound for minimizing the total weighted completion time of coflows in datacenters},
	author       = {Shafiee, Mehrnoosh and Ghaderi, Javad},
	year         = {2018},
	journal      = {IEEE ACM Trans. Netw.},
	publisher    = {Institute of Electrical and Electronics Engineers (IEEE)},
	volume       = {26},
	number       = {4},
	pages        = {1674--1687},
	doi          = {10.1109/TNET.2018.2845852},
	copyright    = {https://ieeexplore.ieee.org/Xplorehelp/downloads/license-information/IEEE.html}
}

@article{Shannon1949,
	title        = {A theorem on coloring the lines of a network},
	author       = {Shannon, Claude},
	year         = {1949},
	journal      = {J. Math. Phys.},
	publisher    = {Wiley},
	volume       = {28},
	number       = {1-4},
	pages        = {148--152},
	copyright    = {http://onlinelibrary.wiley.com/termsAndConditions\#vor}
}

@inproceedings{Spark2010,
	title        = {Spark: cluster computing with working sets},
	author       = {Zaharia, Matei and Chowdhury, Mosharaf and Franklin, Michael J. and Shenker, Scott and Stoica, Ion},
	year         = {2010},
	booktitle    = {Proceedings of the 2nd USENIX Conference on Hot Topics in Cloud Computing},
	location     = {Boston, MA},
	publisher    = {USENIX Association},
	address      = {USA},
	pages        = {10},
	numpages     = {1}
}

@article{Tardos1986,
	title        = {A Strongly Polynomial Algorithm to Solve Combinatorial Linear Programs},
	author       = {Tardos,  {\'{E}}va},
	year         = {1986},
	journal      = {Operations Research},
	publisher    = {Institute for Operations Research and the Management Sciences (INFORMS)},
	volume       = {34},
	number       = {2},
	pages        = {250–256},
	doi          = {10.1287/opre.34.2.250},
	issn         = {1526-5463}
}

@article{toft2021_edgecoloring,
	title        = {A brief history of edge-colorings - with personal reminiscences},
	author       = {Toft, Bjarne and Wilson, Robin},
	year         = {2021},
	journal      = {Discrete Mathematics Letters},
	volume       = {6},
	pages        = {38--46},
	doi          = {10.47443/dml.2021.s105},
	issn         = {2664-2557}
}

@article{Vizing1964,
	title        = {On an estimate of the chromatic class of a p-graph},
	author       = {Vizing, Vadim},
	year         = {1964},
	journal      = {Diskret. Analiz},
	volume       = {3},
	pages        = {25--30}
}

\section{Graph Theory and Edge Colorings}\label{sec:app:graph_theory}
In this section we briefly review relevant graph edge coloring theory, which is used to turn edge to block assignments into valid schedules.
We use $\chi'(G)$ to refer to the chromatic index of a graph $G$, so the minimum number of colors required to properly edge color $G$.
\subsection{Density of a Graph}\label{sec:density}
In our case density refers to edge colorings of the graph and should not be confused with the more common notion of density which relates the number of vertices and number of edges in a graph.

Consider a graph $G=(V,E)$ and a valid edge coloring $c:E \rightarrow [k]$. It is clear by definition that every one of the $k$ color classes $E_1,\dots,E_k$ has to form a matching in $G$, so it immediately follows that $|E_i| \le \lfloor\tfrac{1}{2}|V|\rfloor$. The same relation has to hold for every subgraph $H \subseteq G$ with at least $2$ vertices, meaning that $|E_i[H]| \le \lfloor\tfrac{1}{2}|V[H]|\rfloor$. Because the color classes in a valid coloring form a partition of the edge set, it follows that $|E[H]| \le k \cdot \lfloor\tfrac{1}{2}|V[H]|$. As this has to hold for any subgraph with at least $2$ vertices, by rearranging we obtain a lower bound on $k$, or more specifically a lower bound on the chromatic index of $G$. This is called the density $\mathcal{W}(G)$:
\begin{equation*}
    \mathcal{W}(G) \qeq \max_{H \subseteq G: |V[H]| \ge 2} \left\lceil \frac{|E[H]|}{\lfloor \tfrac{1}{2}|V[H]|\rfloor}\right\rceil
\end{equation*}
As just observed, for any graph $G$ there holds $\mathcal{W}(G) \le \chi'(G)$. Note that sometimes $\Gamma$ is used in place of $\mathcal{W}$ to denote the density. For $\mathcal{W}(G) \ge \Delta(G)$, the density of a general multigraph can be determined in polynomial time \cite{Schrijver2002_combopt}. For $\mathcal{W}(G) < \Delta(G)$ this is open, but this case is not relevant to our analysis.

\subsection{LP Bounds and Density}\label{sec:lpbound_density}
As shown in \cref{sec:iter_lp_round}, we can find an integral point which is feasible for \cref{lp:main_gg} with some slightly weakened constraints. More specifically, for some appropriate choice of $k,\tau \in \mathbb{N}$, we can find a solution $x$ such that for every $b \in [B]$ and $U \subseteq V$, the following constraint is valid:
\begin{equation*}
    \sum_{e \in E[U]}x_e \qle (k_b\cdot \tau + 2(k-1))\cdot \lfloor \tfrac{1}{2}|U|\rfloor \hop{-}1
\end{equation*}
Let $G_{b,x}$ be the graph obtained by only including the edges of block $b$ which are integral in the LP solution. By rearranging above equation, we obtain that $\mathcal{W}(G_{b,x}) \le k_b\cdot \tau + 2(k-1)$. From the bound on the vertex constraint violation, we additionally know that $\Delta(G_{b,x}) \le k_b \cdot \tau  + k - 2$. So this provides upper bounds for both the maximum degree and the density of the induced graph in each block.
For the choice $k=8$ from \cref{lemma:iter_round_violation_bound}, this yields $\Delta(G_{b,x}) \le k_b \cdot \tau + 6$ and $\mathcal{W}(G_{b,x}) \le k_b \cdot \tau + 14$. 

\subsection{Efficiently Realizable Upper Bounds on the Chromatic Index}\label{sec:eff_realizable_bounds}
Determining the exact chromatic index $\chi'(G)$ is an $\algo{NP}$-hard problem and in fact it is even hard to approximate strictly better than a factor of $\frac{4}{3}$, assuming $\algo{P} \neq \algo{NP}$ \cite{holyer1981_edgecolorhard}. While for simple graphs there exists an algorithm which always returns a coloring using at most $\Delta + 1$ colors, for multigraphs the situation is more complicated, as discussed in \cref{sec:intro:ggextension}.\\

We call a function $f$ which maps graphs to integers an efficiently realizable upper bound if for any graph $G$ we can find an integral edge coloring using at most $f(G)$ colors in time polynomial in $|V|$ and $|E|$. So for simple graphs, $f(G) = \Delta(G)+1$ is an efficiently realizable upper bound and for general graphs, $f(G) = \tfrac{3}{2}\Delta(G)$ is another example. This concept is closely related to the well known Goldberg-Seymour conjecture, which states that for a general multigraph $G$ it holds that $\chi'(G) \le \max\{\Delta(G)+1\ ,\ \mathcal{W}(G) \}$.
In recent work, the authors of \cite{Chen2025_goldberg} gave a non-constructive proof of the Goldberg-Seymour conjecture, with one of the authors later also providing an algorithm realizing this bound \cite{jing2024edgecoloringmultigraphs}. Thus, given a general multigraph $G$, $f(G) := \max\{\Delta(G)+1\ ,\ \mathcal{W}(G) \}$ is an efficiently realizable upper bound. As $\Delta(G)$ and $\mathcal{W}(G)$ are both lower bounds on the chromatic index, this result implies that one can in polynomial time find a coloring using at most one additional color compared to the optimum. From the $\algo{NP}$-hardness of determining the chromatic index it follows that this result is best possible.
Note that the algorithm in \cite{jing2024edgecoloringmultigraphs} is currently not peer-reviewed yet. However, there are long established slightly weaker constructive guarantees, such as contained in \cite{scheide2009_edgecolouring}, which could be applied instead and only worsen the resulting guarantees by a bit. They show that $f(G) = \max\{\lfloor\frac{15}{14}\Delta(G) + \frac{12}{14}\rfloor,\ \mathcal{W}(G)\}$ and $f(G) = \max\{\Delta(G),\mathcal{W}(G)\} + \sqrt{\tfrac{1}{2}\max\{\Delta(G),\mathcal{W}(G)\}}$ are efficiently realizable upper bounds. For all relevant values of $\Delta(G)$ and $\mathcal{W}(G)$, the minimum of the two terms is only worse than the Goldberg-Seymour bound by at most a small additive constant. In \cref{lemma:alg_delay_bound}, the increase of each block's size through the procedure is upper bounded by $14$, so a small additional additive loss through a potential worse edge coloring only slightly increases this and thus does not affect approximation guarantees much. In fact, for the vertex degree bounds and density bounds derived in \cref{lemma:iter_round_violation_bound}, for the choices of $\tau$ in \cref{theorem:main_dm_result} and \cref{theorem:graph_scheduling_approx} the weaker guarantees yield the same results, as the density term still dominates. The asymptotic guarantee of $(2+\epsilon)$ is specifically also not affected at all and only requires a bit larger choice of $\tau$ in the corresponding application of $\algo{GBF}^\tau$ to recover the result.

\subsection{Delay Bound}

Using the results of \cref{sec:lpbound_density} and \cref{sec:eff_realizable_bounds}, we describe how to find an edge to time slot assignment given the block assignment and show upper bounds on the delay. We first consider each block separately, before lifting these partial results to overall bounds.

\subsubsection*{Block Size}

For some fixed $b \in [B]$ and an integral solution $x$ as provided by the rounding procedure from \cref{sec:iter_lp_round}, consider the subgraph $G_{b,x}$ consisting of the edges assigned to block $b$. From \cref{lemma:iter_round_violation_bound}, we know that $\Delta(G_{b,x}) \le k_b\cdot\tau +6$ and $\mathcal{W}(G_{b,x}) \le k_b\cdot\tau +14$. From \cref{sec:eff_realizable_bounds} it follows that in polynomial time we can find a coloring of $G_{b,x}$ using at most
\begin{equation*}
\max\{\Delta(G_{b,x}) +1\ ,\  \mathcal{W}(G_{b,x})\}
\end{equation*}
colors. To obtain worst case guarantees it suffices to focus on the second term, as it will strictly dominate the first term under worst case assumptions.

\subsubsection*{Combined Guarantees}
Combining all the previous results leads to the following delay guarantee, where $\algo{GBF}(C_j)$ is the finishing time of the $j$-th job in the schedule returned by the algorithm. The analysis is similar to the one seen in \cite{rohwedder2025_coflow}.

\algdelaybound*

\begin{proof}
The algorithm works by first rounding up all deadlines $C_1,\dots, C_B$ to the next multiple of $\tau$ to obtain the deadlines $\bar{C}_1, \dots \bar{C}_B$. As before, for $b \in [B]$, let the block sizes be given by $k_b \cdot \tau$, for $k_b \in \mathbb{N}_{\ge1}$ and we assume without loss of generality that no two rounded deadlines coincide. Using this construction, from \cref{lemma:iter_round_violation_bound} we obtain that we can find an integral edge to block assignment where each block $G_{b,x}$ has degree at most $k_b \cdot\tau + 6$ and density at most $k_b \cdot\tau + 14$. Hence, using the results from \cref{sec:eff_realizable_bounds}, we can also find a coloring, so an edge to time slot assignment, using at most $k_b \cdot\tau + 14$ time slots.
For some fixed $j \in [B]$, define $a \in [0,\tau)$ and $k \in \mathbb{N}$ such that $C_j = k\cdot \tau + a$. Assume $a > 0$ for now, then $\bar{C}_j = (k+1) \cdot \tau$. The size of each previous block and the block belonging to $C_j$ increases to at most $\tau + 14$, so the final time slot of the block after the rounding and the slot assignment is latest at time point $j \cdot (\tau+14)$. As each block has size at least $\tau$, we have $j \le k+1$ and thus $j\cdot (\tau+14) \ \le\  (k+1)\cdot (\tau+14)$. So we obtain:
\begin{equation*}
    \algo{GBF}^\tau(C_j) \hle (k+1)\cdot (\tau+14)\heq \mfrac{(\tau+14)}{\tau} \cdot \tau \cdot k + \tau + 14 \hle \mfrac{\tau+14}{\tau}C_j + \tau + 14
\end{equation*}
For $a = 0$ the derivation is analogously, just with a stronger starting upper bound on the finishing time.
\end{proof}

\section{2-Approximate Deadlines}\label{sec:app:2appround}
In order to employ \cref{lp:main_gg} and the related rounding results algorithmically, deadlines with good cost guarantees are required. By a generalization from the so called Concurrent Open Shop Scheduling, it is known that the deadlines cannot be approximated strictly better than a factor of $2$, assuming $\algo{P} \neq \algo{NP}$ \cite{Sachdeva2013}. Algorithms to determine such deadlines have been studied in various contexts and by combining two partial results, we obtain the following bound which shows that tight approximations are possible and that we can even gain a slight additive improvement over the guarantee.

\begin{lemma}\label{lemma:app:two_approx}
For a given instance of Graph Scheduling, there is a polynomial time algorithm determining deadlines $C_1,\dots,C_B$ for which \cref{lp:main_gg} is feasible and for which the following cost bound holds with respect to the optimal cost $\algo{OPT}$ of the underlying instance:
\begin{equation*}
\sum_{j \in [B]} \omega_j C_j \qle  2\cdot \algo{OPT} \hop{-}  \sum_{j \in [B]}\omega_j
\end{equation*}
\end{lemma}
\begin{proof}
The proof follows by combining two known results from the literature. In \cite{megow2025groupschedule} they show that for any downward closed polytope, using a randomized rounding scheme, one can find $(2+\epsilon)$-approximate deadlines such that the corresponding LP is feasible. The additive $\epsilon$ can be omitted in case all entries are polynomially bounded. The \cref{lp:main_gg} fits into their setting, as the degree and odd-set constraints are downward closed, hence we can obtain such $2$-approximate deadlines. The polynomial time solvability of the involved LP follows from the well known separation procedure for odd-set inequalities. The same rounding idea has also been used in previous algorithms for Coflow Scheduling \cite{im19,fukunaga22,rohwedder2025_coflow}. In \cite{rohwedder2025_coflow} the analysis of the rounding procedure in the context of Coflow Scheduling has been improved, yielding the additive $- \sum_{j \in [B]}\omega_j$ improvement. This proof directly transfers to our general version, as no properties exclusive to the more restricted case are used.
\end{proof}

\section{Improved Graph Scheduling Approximation}\label{sec:improved_graph_sched_approx}
Using the allocation guarantee from \cref{lemma:alg_delay_bound} together with guarantees from the literature, \cref{theorem:graph_scheduling_approx} can be shown, which gives a slight improvement on the previously best known approximation ratio of $4$ for Graph Scheduling. Note that as mentioned before, this ratio is not optimized and should be straightforward to improve using similar ideas as contained in \cite{rohwedder2025_coflow}.

\graphschedulingapprox*

\begin{proof}
We pick two algorithms which we run independently and return the schedule with lower cost.
From the \hyperlink{link:main_scheduling_alg_proof}{proof of \cref*{thm:main_scheduling_alg}} , we know that we can find a solution $S_1$ with cost:
\begin{equation*}
    \omega(S_1) \qle 2\cdot\frac{\tau+14}{\tau}\cdot \algo{OPT} \hop{+} (\tau+14 - \tfrac{\tau+14}{\tau})\sum_{j \in [B]}\omega_j
\end{equation*}
From the algorithms for Coflow Scheduling contained in \cite{ahmadi_scheduling_2017, Shafiee2018-ik}, combined with refined analysis from \cite{rohwedder2025_coflow}, it follows that there is an algorithm returning a solution $S_2$ with cost:
\begin{equation*}
    \omega(S_2) \qle 4 \cdot \algo{OPT} \hop{-} 3\sum_{j \in [B]}\omega_j
\end{equation*}
Similar to the proof of \cref{theorem:main_dm_result}, we compare the minimum cost of the two schedules to the cost of a randomized algorithm \algo{RALG} which runs the first algorithm with probability $\lambda_1$ and the second with probability $\lambda_2$.
For $\tau := 34$ and $\lambda_1 := \frac{17}{281}$ and $\lambda_2 := \frac{264}{281}$, this gives the desired bound:
\begin{align*}
    \min\{\omega(S_1)\ ,\ \omega(S_2)\} &\qle \mathbb{E}[\omega(S_{\algo{RALG}})]\\
    ~& \qeq \lambda_1 \omega(S_1) \hop{+} \lambda_2 \omega(S_2)\\
    ~&\qeq \mfrac{17}{281}\Big(\mfrac{48}{17}\cdot\algo{OPT} \hop{+} \big(48 - \tfrac{24}{17}\big)\sum_{j \in [B]}\omega_j\Big) \\
    ~&\phantom{=.}\hop{+} \mfrac{264}{281}\Big(4 \cdot \algo{OPT} \hop{-} 3\sum_{j \in [B]}\omega_j\Big)\\
    ~&\qeq \mfrac{1104}{281} \cdot \algo{OPT} \qedhere
\end{align*}
\end{proof}

\section{Local Ratio for Data Migration}\label{sec:app:datamigration}
In this section we give a brief overview over the ideas and techniques used in Local Ratio proofs and more specifically show how the technique has been applied to obtain the so far best approximation ratio of $1+\phi$ ($\approx 2.618$) for Data Migration in \cite{mestre2010}. More comprehensive background information and additional examples can be found in works such as \cite{Bar-Yehuda1985-se,Bar-Yehuda2004-vk}.
\subsection{Local Ratio Concept}
Assume that we are given some maximization or minimization problem with a weight vector $\omega \in \mathbb{R}^n_{\ge 0}$. We present the idea for the minimization setting, the maximization case follows analogously. Assume that we have some decomposition of $\omega$ into vectors $\omega_1,\dotsc,\omega_{\ell} \in \mathbb{R}^n_{\ge 0}$, meaning that $\omega = \sum_{j \in [\ell]}\omega_j$. Let $S$ be some feasible solution to our minimization problem and let $\algo{OPT}$ be an optimum solution. If for some $\alpha \ge 1$, we can show that for all $j \in [\ell]$, the inequality $\omega_j(S) \le \alpha \cdot \omega_j(\algo{OPT})$ is fulfilled, then it follows that $\omega(S) \le \alpha \cdot \omega(\algo{OPT})$, which shows that $S$ is $\alpha$-approximate. This basic observation is a direct consequence of the linearity of the objective function. This result is also known as the Local Ratio Theorem. The key in Local Ratio algorithms is in good construction of the weight decomposition. One common approach is an iterative one, where the weight decomposition is constructed step by step, with other algorithmic decisions taking place in-between. It is often much easier to show approximation guarantees for the individual weight functions, as they are based on local structure or by construction fulfill certain strong properties. For example, in the following section on Data Migration, every individual weight function will be supported on the neighborhood of some vertex, which then enables strong local ratio bounds to be shown.

\subsection{Local Ratio in Data Migration}
We review the algorithm and most important analysis of the Local Ratio approach employed by Mestre \cite{mestre2010} to obtain a $(1+\phi)$-approximation for Data Migration. For full details and analysis we refer to their work.

In Data Migration, we are given a graph $G=(V,E)$ and weights $\omega_v \ge 0$ for $v \in V$. In each timestep, a set of edges which forms a matching in the graph can be scheduled. A vertex finishes if all of its adjacent edges have been scheduled. The goal is to minimize the weighted sum of completion times, meaning that we want to minimize $\sum_{v \in V}\omega_v C_v^*$, where $C_v^*$ is the finishing time of vertex $v \in V$. In terms of Graph Scheduling, this is equivalent to the groups consisting of exactly the edges adjacent to each vertex.

The authors present the \cref{alg:dmlr} for Data Migration. The idea is to assign every vertex some label and then schedule the edges greedily by their labels. Here, for some vertex $u \in V$, $\algo{UN}(u)$ is the set of unlabeled neighbors of the vertex $u$. The algorithm in this form is actually a class of algorithms, as it depends on a hitherto unspecified choice of weight vector in line $7$.

\begin{algorithm}[ht!]
    \DontPrintSemicolon
    \SetAlFnt{10pt}
	\textbf{Algorithm (\cite{mestre2010})} $\algo{DMLR}(V,E,\omega)$:\\
    // Labeling Stage\\
    \ForEach{$v \in V$}
    {
        $\ell_v \gets \underline{\phantom{-}}$
    }
    \Repeat{$\text{every vertex is labeled}$}
    {
        \text{choose } $u \in V$ \text{ maximizing } $\Delta = |\text{UN}(u)|$\\
        \text{choose } $\hat{\omega}$ \text{ with support in UN}$(u)$\\
        $\omega \gets \omega - \min\left\{\frac{\omega(u)}{\hat{\omega}(u)}\middle|\ \hat{\omega}(u)>0\right\}\hat{\omega}$\\
        \ForEach{$v \in \text{UN}(v)\ |\  \omega(v) = 0$}
        {
            $\ell_v \gets \Delta$
        }
    }
    // Scheduling Stage\\
    \text{sort } $(u,v) \in E$ \text{ in lexicographic order of } $\langle \min\{\ell_u,\ell_v\},\max\{\ell_u,\ell_v\} \rangle$\\
    $S \gets \text{ empty schedule}$\\
    \ForEach{$e \in E \text{ in sorted order}$}
    {
        add $e$ to $S$ as early as possible
    }
	\Return $S$
	\caption{The base local ratio algorithm for Data Migration as given in \cite{mestre2010}. The weight vector $\hat{\omega}$ in line $7$ is calculated through an LP, which we do not include in this pseudocode for easier readability.}
	\label{alg:dmlr}
\end{algorithm}

The vector $\hat{\omega}$ in line $7$ is chosen as the minimizer of the following LP, where $d_i$ is the degree sequence of the vertices in $\algo{UN}(u)$:

\begin{gather*}\tag{\textbf{LP M}}\label{lp:mestre_local_ratio}
\begin{aligned}
~&&&\min \quad\sum_{i \in [\Delta]}(d_i + \Delta - 1)\hat{\omega}_i\\
\sum_{i \in [\Delta]}(y_i - z_i) &&\ge& \quad 1 &&\\
y_i - z_j &&\le&\quad \max\{d_i,j\}\hat{\omega}_i && \forall i,j\in [\Delta]\\
y_i,z_i,\hat{\omega}_i &&\ge&\quad 0 &&\forall i \in [\Delta]
\end{aligned}
\end{gather*}
In their analysis, for a given degree sequence $d$ and a weight function $\omega$, they show upper and lower bounds on the weighted sum of completion times and define several constants based on these. The upper bound $\algo{UB}(d,\omega)$ is given as follows
\begin{equation*}
    \algo{UB}(d,\omega) \qeq \sum_{i \in [\Delta]} \omega_i (d_i + \Delta - 1),
\end{equation*}
while the lower bound is as follows:
\begin{equation*}
    \algo{LB}(d, \omega) \qeq \min_{\sigma:[\Delta] \rightarrow [\Delta]} \sum_{i \in [\Delta]} \omega_i \max\{\sigma(i),d_i\}
\end{equation*}
Here $\sigma:[\Delta] \rightarrow [\Delta]$ is the set of all permutations on $[\Delta]$. They define $\rho(d) := \inf_{\hat{\omega}} \frac{\algo{UB}(d,\hat{\omega})}{\algo{LB}(d,\hat{\omega})}$ and call this the minimum local ratio of $d$. The \cref{lp:mestre_local_ratio} aims to minimize the local ratio under the assumption $\algo{LB}(d,\omega) = 1$, which is valid as the ratio is scaling invariant. The worst case local ratio is defined as $\rho := \sup_d \rho(d)$, with the worst case ratio for degree sequences of length $\Delta$ being written as $\rho_{\Delta} := \max_{d:|d| = \Delta} \rho(d)$.

\subsection{Bounds on the Local Ratio}
In their work, the authors of \cite{mestre2010} show several results for the local ratio of \cref{alg:dmlr} for Data Migration, which leads to the $(1+\phi)$-approximation guarantee of the algorithm. In this section we review their results and subsequently provide more fine-grained guarantees by strengthening and extending their analysis. The following three results can be found in slightly modified form in \cite{mestre2010}.

\begin{lemma}[\cite{mestre2010}]\label{lemma_mestre_1}
For any $\Delta \in \mathbb{N}_{\ge 2}$, the value of $\rho_\Delta$ is upper bounded by $(1+\phi) + \frac{2}{\Delta-1}$.
\end{lemma}

\begin{lemma}[\cite{mestre2010}]\label{lemma_mestre_2}
For any $\Delta \in \mathbb{N}_+$, there holds $\rho_\Delta < \rho_{2\Delta}$.
\end{lemma}
\noindent Combining these two results leads to the $(1+\phi)$-approximation guarantee.

\begin{lemma}[\cite{mestre2010}]\label{lemma_mestre_3}
The local ratio $\rho$ is upper bounded by $(1+\phi)$.
\end{lemma}

Keeping in line with the idea behind Local Ratio, we show that this bound can be additively improved on a local scale, which then directly lifts to a general additive improvement.

\begin{lemma}\label{lemma_local_ratio_bound}
Let $\omega^{\algo{ALG}}$ be the cost function derived in step $7$ in \cref{alg:dmlr} and let $S$ be any schedule and $\algo{OPT}$ be an optimal schedule. Then we can bound the local ratio in the following way:
\begin{equation*}
    \omega^{\algo{ALG}}(S) \qle (1 + \phi)\cdot\omega^{\algo{ALG}}(\algo{OPT}) \hop{-} \mfrac{1}{2} \sum_{j \in [\Delta]}\omega^{\algo{ALG}}_j
\end{equation*}
\end{lemma}
\begin{proof}
We start by deriving a sharper version of \cref{lemma_mestre_2} using their proof as a base. Let $d$ be a degree sequence of length $\Delta$ such that $\rho_\Delta = \rho(d)$ and let $k \in \mathbb{N}_{\ge 2}$. We define a new degree sequence $d'$ of length $k\Delta$ which we use to derive a first bound. For this purpose, replace any element in $d$ by $k$ copies of it which are multiplied by $k$ to obtain $d'$, i.e. the first $k$ elements in $d'$ have value $kd_1$, the next $k$ value $kd_2$ and so on. Solve the \cref{lp:mestre_local_ratio} for $d'$ to obtain the values $(\omega',y',z')$. For $i \in [\Delta]$, define $y_i := \sum_{j \in [i\cdot k-(k-1),...,i\cdot k]}y'_j$, $z_i := \sum_{j \in [i\cdot k-(k-1),...,i\cdot k]}z'_j$, and $\omega_i = k \cdot \sum_{j \in [i\cdot k-(k-1),...,i\cdot k]}\omega'_j$. Before we show that $(\omega,y,z)$ is feasible for \cref{lp:mestre_local_ratio} on $d$, we derive the following bound:
\begin{align*}
\quad\sum_{i \in [\Delta]} (d_i + \Delta - 1)\omega_i &\qeq\sum_{i \in [\Delta]}(d_i + \Delta - 1)\cdot k \cdot \sum_{j \in [i\cdot k-(k-1),...,i\cdot k]}\omega'_j\\
~&\qeq\sum_{i \in [k\Delta]}(kd_i + k\Delta - k)\cdot\omega'_i\\
~&\qeq\sum_{i \in [\Delta']}(d_i' + \Delta' - k)\cdot \omega_i'\\
~&\qeq\sum_{i \in [\Delta']}(d_i' + \Delta' - 1)\cdot \omega'_i \hop{-} (k-1)\sum_{i \in [\Delta']}\omega_i'\\
~&\qeq\sum_{i \in [\Delta']}(d_i' + \Delta' - 1)\cdot \omega'_i\hop{-} \mfrac{k-1}{k}\sum_{i \in [\Delta]}\omega_i\\
~&\qle\rho_{k\Delta} \hop{-} \mfrac{k-1}{k}\sum_{i \in [\Delta]}\omega_i
\end{align*}
From the choice of $d$ and the bound from \cref{lemma_mestre_3} this implies
\begin{equation*}
    \rho_\Delta \qle \rho_{k\Delta} \hop{-} \mfrac{k-1}{k}\sum_{i \in [\Delta]}\omega_i \qle (1+\phi) \hop{-} \mfrac{k-1}{k}\sum_{i \in [\Delta]}\omega_i.
\end{equation*}
As $k$ can be chosen arbitrarily large, it in fact follows that we can drop the fraction $\frac{k-1}{k}$ to obtain
\begin{equation*}
    \rho_\Delta \qle (1+\phi) \hop{-} \sum_{i \in [\Delta]}\omega_i.
\end{equation*}
The $\omega_i$ in the previous bound is the weight vector derived from the optimal solution to \cref{lp:mestre_local_ratio} for degree sequence $d'$. For our purposes we require a bound against $\omega^{\algo{LP}}$, the optimal solution for degree sequence $d$. We have:
\begin{align*}
    1 &\qge \frac{\sum_{i \in [\Delta]}(d_i + \Delta - 1)\omega_i^{\algo{LP}}}{\sum_{i \in [\Delta]}(d_i + \Delta-1)\omega_i} \qeq \frac{\sum_{i \in [\Delta]}(\frac{d_i + \Delta - 1}{2\Delta})\omega_i^{\algo{LP}}}{\sum_{i \in [\Delta]}(\frac{d_i + \Delta-1}{2\Delta})\omega_i} \\
    ~&\qge \frac{\frac{1}{2} \cdot \sum _{i \in [\Delta]}\omega_i^{\algo{LP}}}{\frac{2\Delta-1}{2\Delta}\cdot \sum_{i \in [\Delta]} \omega_i} \qeq \frac{\Delta}{2\Delta-1} \cdot \frac{\sum_{i \in [\Delta]} \omega_i^{\algo{LP}}}{\sum_{i \in [\Delta]} \omega_i}
\end{align*}
In the third step we use the fact that all $d_i$ are in $[\Delta]$, so $\frac{d_i+\Delta-1}{2\Delta} \in \left[\frac{1}{2},\frac{2\Delta-1}{2\Delta}\right]$.
Thus by rearranging we obtain $\sum_{i \in [\Delta]} \omega_i \Hquad \ge \Hquad \frac{\Delta}{2\Delta-1}\sum_{i \in [\Delta]} \omega_i^{\algo{LP}}$.\\

\noindent We now show that the triple $(\omega,y,z)$ is feasible for \cref{lp:mestre_local_ratio}. We have
\begin{equation*}
    \sum_{i \in [k\Delta]}(y'_i - z'_i) \hge 1 \quad \Leftrightarrow \quad \sum_{i \in [\Delta]}(y_i - z_i) \hge 1
\end{equation*}
by simple combination of the terms, so the first set of constraints in \cref{lp:mestre_local_ratio} is fulfilled. For any $\ell \in [k]$ and $i,j \in \{\ell \cdot k - (k-1),\dotsc,\ell\cdot k\}$ we know that 
\begin{equation*}
    y'_i - z'_j \hle \max\{d'_i,j\}\omega'_i \heq \max\{kd_\ell,j\} \omega'_i
\end{equation*}
holds. By summing these inequalities over all values in the index set, we obtain:
\begin{equation*}
    y_i - z_j \heq\sum_{\kappa \in \{0,\dotsc,k-1\}} (y'_{i\cdot k - \kappa} - z'_{j\cdot k - \kappa}) \hle \sum_{\kappa \in \{0,\cdots,k-1\}} \max\{kd_{i},j\cdot k - \kappa\}\omega'_{i\cdot k - \kappa} \hle \max\{d_i,j\}\omega_i
\end{equation*}
So the second set of constraints is also fulfilled, which shows that $(\omega,y,z)$ is feasible.\\

We proceed by bounding the actual cost of any schedule in terms of the algorithm's objective function. Let $\omega^{\algo{LP}}$ be the optimal LP solution calculated in the \cref{alg:dmlr} and let $\omega^{\algo{ALG}}$ be the scaled weight function used by the algorithm. Let $\alpha > 0$ be such that $\omega^{\algo{LP}} \cdot \alpha = \omega^{\algo{ALG}}$. We can derive the following relation for the lower bound $\algo{LB}(d,\omega^{\algo{LP}})$:
\begin{align*}
    \algo{LB}(d,\omega^{\algo{LP}}) &\heq \max_{\sigma: [\Delta] \rightarrow [\Delta]} \sum_{i \in [\Delta]}\omega^{\algo{LP}}_i \max\{d_i, \sigma(i)\}\\ ~&\heq \max_{\sigma: [\Delta] \rightarrow [\Delta]} \sum_{i \in [\Delta]}\mfrac{1}{\alpha}\cdot\omega^{\algo{ALG}}_i \max\{d_i, \sigma(i)\} \heq \mfrac{1}{\alpha}\algo{LB}(d,\omega^{\algo{ALG}})
\end{align*}
Using the previous two inequalities, the cost $\omega(S)$ can be bounded in the following way:
\begin{align*}
\omega^{\algo{ALG}}(S) \qquad &\le \qquad \algo{UB}(d,\omega^{\algo{ALG}}) \quad = \quad \frac{\algo{UB}(d,\omega^{\algo{ALG}})}{\algo{LB}(d,\omega^{\algo{ALG}})} \cdot \algo{LB}(d,\omega^{\algo{ALG}}) \quad = \quad \rho(d) \cdot \algo{LB}(d,\omega^{\algo{ALG}}) \qquad \\
&\le \qquad \Big(\rho_\infty - \mfrac{\Delta}{2\Delta-1}\sum \omega^{\algo{LP}}_i\Big) \cdot \algo{LB}(d,\omega^{\algo{ALG}})\\
~&\le\qquad(1+\phi)\algo{LB}(d,\omega^{\algo{ALG}}) \hop{-}  \left(\mfrac{\Delta}{2\Delta-1}\sum \mfrac{1}{\alpha}\omega^{\algo{ALG}}_i\right) \cdot \alpha \cdot \algo{LB}(d,\omega^{\algo{LP}}) \\
~&\stackrel{\mathclap{\algo{LB}(d,\omega^{\algo{LP}}) \ge 1}}{\le}\qquad (1+\phi)\algo{LB}(d,\omega^{\algo{ALG}}) \hop{-} \mfrac{\Delta}{2\Delta-1}\sum \omega^{\algo{ALG}}_i
\end{align*}
In the third step we use that $\omega^{\algo{ALG}}$ is a scalar multiple of $\omega^{\algo{LP}}$ and thus $\frac{\algo{UB}(d,\omega^{\algo{ALG}})}{\algo{LB}(d,\omega^{\algo{ALG}})} = \frac{\algo{UB}(d,\omega^{\algo{LP}})}{\algo{LB}(d,\omega^{\algo{LP}})}$ and as the algorithm chooses the $\omega^{\algo{LP}}$ which minimizes this ratio, we have $\frac{\algo{UB}(d,\omega^{\algo{LP}})}{\algo{LB}(d,\omega^{\algo{LP}})} = \inf_{\hat{\omega}} \frac{\algo{UB}(d,\hat{\omega})}{\algo{LB}(d,\hat{\omega})} = \rho(d)$. The claimed bound follows from the observations $\algo{LB}(d,\omega^{\algo{ALG}}) \le \omega^{\algo{ALG}}(\algo{OPT})$ and $\frac{\Delta}{2\Delta-1} > \frac{1}{2}$. 

\end{proof}
From the Local Ratio Theorem and \cref{lemma_local_ratio_bound} it follows directly that for any weight vector $\omega$, the \cref{alg:dmlr} returns a solution which also fulfills a bound of this general form:
\localratioimprov*

\end{document}